%% file: defrag.tex
\documentclass[oribibl,runningheads]{llncs}
\usepackage{amssymb}
\usepackage{graphicx}
\usepackage{color}

\usepackage{url}
\urldef{\mailsa}\path|{s.fekete, n.schweer}@tu-bs.de, tom@kamphans.de |
\urldef{\mailsb}\path|| \urldef{\mailsc}\path||

\usepackage[linesnumbered,boxed,noend]{algorithm2e}
\usepackage{amsmath}
\usepackage{epsfig}

\DeclareMathOperator*{\argmax}{argmax}

\newtheorem{lem}{Lemma}
\newtheorem{thm}{Theorem}

\newcommand{\N}{\mathbb{N}}

 \usepackage{changebar}

\iffalse
 \usepackage[normalem]{ulem}
 \def\geaendert#1{\uwave{#1}}
 \def\eingefuegt#1{\uline{#1}}
 \def\gestrichen#1{\xout{#1}}
\else
 \def\geaendert#1{{#1}}
 \def\eingefuegt#1{{#1}}
 \def\gestrichen#1{}
\fi

\begin{document}

\mainmatter  

\title{Maintaining Arrays of Contiguous Objects}
\titlerunning{Maintaining Arrays of Contiguous Objects}

\author{Michael A. Bender\inst{1} \and S\'andor P. Fekete\inst{2}
\and Tom Kamphans\inst{2}\thanks{Supported by DFG grant FE 407/8-3, project
``ReCoNodes''} \and Nils Schweer\inst{2}}

\institute{
Department of Computer Science
State University of New York at Stony Brook
Stony Brook, NY 11794-4400, USA
\and
Braunschweig University of Technology,
Department of Computer Science,
Algorithms Group,
38106 Braunschweig, Germany
}

\toctitle{Maintaining Arrays of Physical Objects}

\maketitle

\begin{abstract}
In this paper we consider methods for dynamically storing a
set of different objects (``modules'') in a physical array.
Each module requires
one free {\em contiguous} subinterval in order to be placed.
Items are inserted or removed, resulting in
a fragmented layout that makes it harder to insert further modules.
It is possible to relocate modules, one at a time,
to another free subinterval that is contiguous and
does not overlap with the current location of the module.
These constraints clearly distinguish our problem from classical
memory allocation.
We present a number of algorithmic
results, including a bound of $\Theta(n^2)$ on physical sorting
if there is a sufficiently large free space
and sum up NP-hardness results for arbitrary initial layouts.
For online scenarios in which modules arrive one at a time,
we present a method that requires $O(1)$ moves per insertion
or deletion and amortized cost $O(m_i \lg \hat{m})$ per insertion or deletion,
where $m_i$ is the module's size, $\hat{m}$ is the size of the largest module
and costs for moves are linear in the size of a module.

\end{abstract}

\section{Introduction} \label{sec:introduction}
Maintaining a set of objects is one of the basic problems
in computer science. As even a first-year student knows,
allocating memory and arranging objects (e.g., sorting or
garbage collection)
should not be done by moving the objects,
but merely by rearranging pointers.

The situation changes when the objects to be sorted
or placed cannot be rearranged in a virtual
manner, but require actual physical moves; this is
the case in a densely packed warehouse,
truck or other depots, where items have to be added
or removed. Similarly, allocating
space in a fragmented array is much harder when
one contiguous interval is required for each object:
Even when there is sufficient overall free space,
placing a single item may require rearranging the
other items in order to create sufficient {\em connected}
free space. This scenario occurs for
the application that initiated
our research: Maintaining
modules on a Field Programmable Gate Array
(FPGA);
reconfigurable chips that consist of
a two-dimensional array of processing units.
Each unit can perform one basic operation depending on its
configuration, which can be changed during runtime.
A module is a configuration for a set of processing units
wired together to fulfill a certain task.
As a lot of FPGAs allow only whole columns to be reconfigured, we allow the
modules to occupy only whole columns on the FPGA (and deal with a
one-dimensional problem).
Moreover, because the layout of the modules
(i.e., configurations and interconnections of the processing units)
is fixed, we have to allocate
{\em connected} free space for a module on the FPGA.
In operation,
different modules are loaded onto the FPGA, executed for some time
and are removed when their task is fulfilled, causing {\em fragmentation}
on the FPGA.
When fragmentation becomes too high (i.e., we cannot place modules,
although there is sufficent free space, but no sufficent amount of
connected free space), the execution of new task has to be delayed
until other tasks are finished and the corresponding modules are
removed from the FPGA. To reduce the delay,
we may reduce fragmentation by moving
modules. Moving a module means to stop its operation, copy the module to an
unoccupied space, restart the module in the new place, and declare the
formerly occupied space of the module as free space; see Figure~\ref{move}.
Thus, it is important that the current and the target position of the module
are {\em not overlapping} (i.e., they do not share a column).
This setting gives rise to two approaches: We may either use simple placing
strategies such as first fit and compact the whole FPGA when necessary
(as discussed in \cite{fkstvakt-nbddr-08}),
or use more elaborated strategies that organize the free space and
avoid the need for complete defragmentations.

\gestrichen{
Only very few physical scenarios allow rearrangement
by continuous ``sliding'' moves, which allows simply pushing
together all present objects in order to create one
contiguous free interval.
Instead, items have to be relocated by a sequence of discrete moves,
one at a time,
to a free, contiguous interval;
see Figure~\ref{move}.
}

\begin{figure}[t]
  \centerline{\epsfig{figure=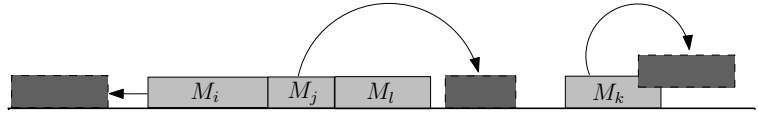,width=100mm}}
\caption{A module corresponds to a set of columns on an FPGA. Each
module occupies a contiguous block of array cells. Module $M_i$ is
shifted and module $M_j$ is flipped. The move of module $M_k$ is
forbidden, because the current and the target position overlap. If these
kind of moves would be allowed connecting the total free space could
always be done by shifting all modules to one side.}\label{move}
\vspace*{-6mm}
\end{figure}

{\bf Related Work.}
There is a large body of work on storage allocation; e.g.,~\cite{k-acpfa-97}
for an overview and \cite{LubyNaorOrda96,NaorEtAl00} for competitive analysis
of some algorithms. Many storage allocation algorithms also have
analogues in bin packing~\cite{CoffmanGareyJohnson96}.  The salient feature of
most traditional memory-allocation and bin-packing heuristics is that
once an item is allocated, it cannot be moved, unlike the model is this paper.
There is also a tremendous amount of work on classical sorting
(see, e.g., \cite{k-acpss-97}).

Physical allocation, where elements can be placed and then
moved, has received less attention.
Itai, Konheim, and Rodeh consider maintaining $n$ unit-size objects
sorted in an $O(n)$ sized array by appropriately maintaining a linear number of
gaps interspersed between the elements at an amortized cost of $O(\lg^{2}n)$ per insert,
and the problem is deamortized in~\cite{Willard92}.
The packed memory array
Bender, Demaine, and Farach-Colton~\cite{bdc-cobt-05} and
Bender and Hu~\cite{BenderHu07} investigate a similar problem in the context of external-memory
and cache-oblivious algorithms.
Bender, Farach-Colton, and Mosteiro~\cite{BenderFaMo06}
show that probabilistically a modified insertion sort runs in $O(n \lg n)$
by leaving appropriate gaps between elements.
In these papers, elements have unit size and there is a fixed order that needs to be
maintained dynamically, unlike the problem in this paper.

A different
problem is described by \cite{gk-pbss-06}, who consider densely
packed physical storage systems for the U.S.\ Navy, based on the classical
15-puzzle, where items can be moved to an adjacent empty cell.
How should one arrange and maintain
the set of free cells, and how can objects be retrieved
as quickly as possible?

Finally, if the sequence of modules (i.e., their size, processing time, and
arrival time) is fully known, then the problem can be stated as a strip
packing problem (without rotation) with release times
for rectangles with widths and heights corresponding to
the module's size and time, respectively.
There is a $(1+\varepsilon)$-approximation for (classical) offline strip packing
\cite{Kenyon96}.
For the case with release times,
Augustine et al.~\cite{Augustine06} give
a $O(\lg n)$ approximation and
a 3-approximation for heights bounded by one.
For approaches from the FPGA community see  \cite{fkstvakt-nbddr-08} and the
references cited in this paper.

{\bf This Paper.}
Dealing with arrangements
of physical objects or data that require contiguous
memory allocation
and nonoverlapping moves
gives rise to a variety of problems
that are quite different from virtual storage management:\\[-18pt]
\begin{itemize}
\item Starting configuration vs.\ full management.
We may be forced to start from an arbitrary configuration,
or be able to control \eingefuegt{the placement of objects.}
\gestrichen{some aspects of the possible
configurations.
This leads to drastically different algorithmic
challenges.}
\item Physical sorting.
Even when we know that it is possible to achieve connected
free space, we may not want to get an arbitrary arrangement
of objects, but may be asked to achieve one in which the objects
are sorted by size.
\item Low-cost insertion.
We may be interested in requiring only a small number
of moves per insertion, either on average, or in the worst case.
\item Objective functions.
Depending on the application scenario, the important
aspects may differ: We may want to minimize the
moves for relocating objects, or the total mass that is moved.
Alternatively, we may perform only very few moves (or none at all),
at the expense of causing waiting time for the objects that
cannot be placed; this can be modeled as minimizing the makespan
of the corresponding schedule.
\end{itemize}

\noindent
{\bf Main Results.}
Our main results are as follows:\\[-18pt]
\begin{itemize}
\item We demonstrate that sorting the modules by size may require
$\Omega(n^2)$ moves.

\item We show that keeping the modules in sorted order
is sufficient to maintain connected free space and to achieve an
optimal makespan, requiring
$O(n)$ moves per insertion or deletion.

\item We give an alternative strategy that guarantees connected
free space; in most steps, this requires $O(1)$ moves for insertion,
but may be forced to switch to sorted order in $O(n^2)$ moves
for high densities.

\item We present an online method that needs $O(1)$ moves per insertion
or deletion.

\item We perform a number of experiments to compare the strategies.

\item For the sake of completeness, we briefly cite 
and  sketch that it is strongly NP-hard to find
an optimal defragmentation sequence when we are forced
to start with an arbitrary initial configuration,
that (unless P is equal to NP) it is impossible
to approximate the maximal achievable free space within
any constant, and
prove that achieving connected space is always possible
for low module density.
\end{itemize}

The rest of this paper is organized as follows. In Section~2,
we introduce the problem and notation. Section~3 discusses aspects
of complexity for a (possibly bad) given starting configuration.
Section~4 focuses on sorting.
Section~5 introduces two insertion strategies that always guarantee
that free space can be made connected. Moreover, we present
strategies that achieve low (amortized or worst-case) cost per
insertion.
Some concluding thoughts are given in Section 6.

\section{Preliminaries} \label{sec:definitions}
Motivated by our FPGA application, we model the problem
as follows:
Let $A$ be an array (e.g., a memory \eingefuegt{or FPGA columns})
that consists of $|A|$ cells. A module $M_i$ of
size $m_i$ occupies a subarray of size $m_i$ in $A$ (i.e., $m_i$ consecutive cells).
We call a subarray of
maximal size where no module is placed a {\em free space}. The
$i$th free space (numbered from left to right)
is denoted by $F_i$ and its size by $f_i$. 

A module located in a subarray, $A_{s}$, can be {\em moved} to another
subarray, $A_{t}$, if $A_t$ is of the same size as $A_{s}$ and all cells
in $A_t$ are empty \geaendert{(particularly, } both subarrays {\bf do not have
a cell in common}). Moves are distinguished into {\em shifts} and {\em flips}:
If there is at least one module located between $A_s$ and $A_t$
we call the move a flip, otherwise a shift; see Fig.~\ref{move}.
Following the two approaches mentioned in the introduction,
we are mainly interested in the following problems.

{\bf Offline Defragmentation:}
We start with a given configuration of modules
in an array $A$ and look for a sequence of moves such that
there is a free space of maximum size. We state the problem
formally:\\
Given: An array $A$, and a set of modules, $M_1,
M_2,...,M_n$, placed in $A$.\\
Task: Move the modules such
that there is a free space of maximum size.

{\bf Online Storage Allocation: } This problem arises from inserting a sequence
of modules, $M_1,M_2,\ldots,M_n$, which arrive in an online fashion, 
the next module
arrives after the previous one has been inserted. After insertion, a
module stays for some period of time in the array before it is removed;
the duration is not known \geaendert{when placing an object}.
\eingefuegt{If an arriving module cannot be placed (because there is
no sufficient connected free space), it has to wait until the array
is compacted or other modules are removed. } The modules in the
array can be moved as described above to create free space for
further insertions.

Our goals are twofold: On the one hand we want to minimize
the {\em makespan} (i.e., the time until the last module
is removed from the array) and, on the other hand, we want to
minimize the costs for the moves. Moves are charged using a
function, $c(m_i)$, which is linear in $m_i$. For example, we can
simply count the number of moves using $c_1(m_i) := 1$, or we count
the moved mass (i.e., we sum up the sizes of the moved modules) with
$c_2(m_i) := m_i$. Formally:\\
Given: An empty array, $A$, a
sequence of modules, $M_1, M_2,...,M_n$, arriving one after the other.\\
Task: Place the modules in $A$ such that (1) the makespan and (2)
the total costs for all moves performed during the insertions is
minimized.

\section{Offline Defragmentation}\label{sec:DefragmentationProblem}
In this section, we assume that we are given an array
that already contains $n$ modules. Our task is to compact the array;
that is, move the
modules such that we end up with one connected free space.
Note that a practical motivation in the context of
dynamic FPGA reconfiguation as well as some heuristics were already
given in our paper~\cite{fkstvakt-nbddr-08}. As they lay the basis
of some of the ideas in the following sections and for the sake of
completeness, we briefly cite and sketch the corresponding
complexity results.

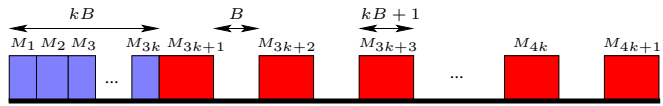
\begin{figure}[b]
\begin{center}
 \input{npcomplete-scaled.pstex_t}
\end{center}
\vspace*{-6mm}
\caption{\label{fig:npcomplete}Reducing 3-Partition to the MDP.}
\vspace*{-6mm}
\end{figure}

\begin{thm}
Rearranging an array with modules $M_1,\ldots,M_n$ and free spaces
$F_1,\ldots, F_k$ such that there is a free space of maximum size is
strongly NP-complete. Moreover,
there is no
deterministic polynomial-time approximation algorithm within any
polynomial approximation factor (unless P=NP).
\end{thm}

The proof is based on a reduction of 3-PARTITION, see Figure~\ref{fig:npcomplete}.
The sizes of the first $3k$ modules correspond to the input of a 3-PARTITION instance,
the size of the free spaces, $B$, is the bound from the 3-PARTITION instance.
We can achieve a free space of maximum size, if and only if we can move the first $3k$
modules to the free spaces, which corresponds to a solution for the 3-PARTITION instance.
The inapproximability argument uses a chain of immobile modules of increasing
size that can be moved once a 3-PARTITION has been found, see~\cite{fkstvakt-nbddr-08}.


This hardness depends on
a number of immobile modules, i.e., on relatively small
free space.
If we define for an array $A$ of length $|A|$ the density
to be $\delta= \frac{1}{|A|}\sum_{i=1}^n m_i$, it is not hard to see
that if
\begin{equation}
\delta \leq \frac{1}{2} - \frac{1}{2|A|} \cdot \max_{i=1,\ldots,n}
\{m_i\}\quad\mbox{\rm or} \label{eq:lowden2}
\end{equation}
\begin{equation}
\max_{i=1,\ldots,n} \{m_i\} \leq \max_{j=1,\ldots,k} \{f_j\}\,.
\label{eq:maxModmaxFsPropertey}
\end{equation}
is fulfilled,
the total free space can always be connected with $2n$ steps by
Algorithm~1 \eingefuegt{which shifts all modules to the right in the
first loop and all modules to the left in the second loop. Starting
at the right and left end, respectively.}

\begin{algorithm}[t]\caption{LeftRightShift}
\label{alg:leftrightshift}


\KwIn{An array $A$ with $n$ modules $M_1,\ldots, M_n$ (numbered from
left to right) such that Eq.~(\ref{eq:lowden2}) or
Eq.~(\ref{eq:maxModmaxFsPropertey}) is fulfilled.}

\KwOut{A placement of $M_1,\ldots, M_n$ such that there is only one
free space. }
\For{$i=n$ \KwTo $1$ } {
        Shift the $M_i$ to the right as far as possible.}
\For{$i=1$ \KwTo $n$ } {
        Shift $M_i$ to the left as far as possible.} 
\end{algorithm}

\begin{thm}\label{thm:totfslowden}
Algorithm \ref{alg:leftrightshift} connects the total free space
with at most $2n$ moves and uses $O(n)$ computing time.
\end{thm}

In the following, we use the idea of Algorithm~1
for maintenance strategies that can accommodate any module for which
there is sufficient total free space.

\section{Sorting} \label{sec:sorting}
In the next section, we present some strategies that are
based on sorting the set of modules by their size. But more than
that, sorting is always an important task. Thus, in this section
we focus on the sorting problem for modules
solely. Note
that we cannot apply classical sorting algorithms such as Quicksort
or Selection Sort, because they assume that every object is of the
same size. We state an algorithm that is similar to
Insertion Sort and show that it can be applied to our setting. It sorts
$n$ modules in an array with $O(n^2)$ steps. Moreover we show that
this is best possible up to a constant factor. More precisely, we
deal with the following problem: Given an array, $A$, with modules
$M_1,\ldots,M_n$ and free spaces $F_1,\ldots,F_k$. Sort the modules
according to their size such that there is only one free space in
$A$.
It is necessary to be able to move every module. Therefore we assume
in this section that Eq.~\eqref{eq:maxModmaxFsPropertey}
is fulfilled in the initial placement. Note that if 
Eq.~\eqref{eq:maxModmaxFsPropertey} is not fulfilled, there are instances for
which it is NP-hard to decide whether it can be sorted or not; this
follows from a similar construction as in 
Section~\ref{sec:DefragmentationProblem}.

\subsection{Sorting {\boldmath $n$} modules with $O(n^2)$ steps}

To sort a given configuration, we first apply
Algorithm~\ref{alg:leftrightshift}, performing $O(n)$ moves.%
\footnote{A short proof of
correctness for this procedure can be found in
\cite{fkstvakt-nbddr-08}.} 
Afterwards, there is only one free space at the right end of $A$ and
all modules are lying side by side in $A$.
We number the modules in the resulting position from left
to right from $1$ to $n$. The algorithm maintains a list $I$ of
unsorted modules. As long as $I$ is not empty, we proceed as
follows: We flip the largest unsorted module, $M_k$, to the right
end of the free space and shift all unsorted modules that were
placed on the right side of $M_k$ to the left. Note that afterwards
there is again only one free space in $A$.


\begin{thm}
Let $A$ be an array with modules $M_1,\ldots,M_n$, free spaces
$F_1,\ldots,F_k$, and let Eq.~(\ref{eq:maxModmaxFsPropertey})
be satisfied. Then Algorithm \ref{alg:sort} sorts the array with
$O(n^2)$ steps.
\end{thm}

\begin{proof}
The while loop is executed at most $n$ times. In every iteration
there is at most one flip and $n$ shifts. This yields an upper bound
of $n^2$ on the total number of moves.

For correctness,
we prove the following invariant: At the end of an iteration of the
while loop, all $M_j$, $j\notin I$, lie side by side at the right
end of $A$ in increasing order (from left to right) and all $M_j$,
$j\in I$, lie side by side at the left end of $A$. We call the first
sequence of modules sorted and the other one non-sorted.

\begin{algorithm}[t]\label{alg:sort}\caption{SortArray}

\SetKw{Kwand}{and}

\KwIn{An array $A$ such that Eq.~(\ref{eq:maxModmaxFsPropertey})
is satisfied.
}

\KwOut{The modules $M_1,\ldots, M_n$ side by side in sorted order
and one free space at the left end of $A$. }

Apply Algorithm ~\ref{alg:leftrightshift}\\
%
$I:=\{1,\ldots,n\}$  \\
\While{ $I \neq \emptyset$ } {
    $k = \argmax_{i\in I} \{m_i\} $\\
    flip $M_k$ to the right end of the free space \nllabel{algline:jump}\\
    $I = I \setminus \{k\}$ \\

    \For{$i=k+1,\ldots,n$ \Kwand $i \in I$}{
            shift $M_i$ to the left as far as possible \nllabel{algline:shift} \\ 
    }
 }
\end{algorithm}


Now, assume that we are at the beginning of the $j$th iteration of
the while loop. Let $k$ be the index of the current maximum in $I$.
By the induction hypothesis and by Eq.~(\ref{eq:maxModmaxFsPropertey}),
the module $M_k$ can be flipped to the only free space. This step
increases the number of sorted elements lying side by side at the
right end of $A$. Since in every step the module of maximum is
chosen, the increasing order in the sequence of sorted modules is
preserved. Furthermore, this step creates a free space of size $m_k$
that divides the sequence of non-sorted modules into two (possible
empty) subsequences. By the numbering of the modules,
the left subsequence contains only indexes smaller than $k$. This
ensures that in the second while loop only modules from the right
subsequence are shifted. Again, since $M_k$ is chosen to be of
maximum size
all shifts are well defined.
At the end of the iteration, the non-sorted modules lie side by side
and so do the sorted ones. \qed
\end{proof}


\subsection{A Lower Bound of $\Omega(n^2)$}

We show that Algorithm \ref{alg:sort} needs the minimum number of
steps (up to a constant factor) to sort $n$ modules. In particular,
we prove that any algorithm needs $\Omega(n^2)$ steps to sort the
following example. The example consists of an even number of
modules, $M_1,\ldots,M_n$, with size $m_i=k$ if $i$ is odd and $m_i
= k+1$ if $i$ is even for a $k\geq 2$. There is only one free space
of size $k+1$ in this initial placement at the left end of $A$, see
Fig.~\ref{fig:lbsort}.

\begin{figure}
\centering
\epsfig{file=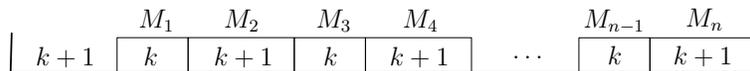, width=10cm}
\caption{\label{fig:lbsort} Sorting an array is in $\Omega(n^2)$.}
\end{figure}

\begin{lem} \label{lem:sort1}
The following holds for any sequence of shifts and flips applied to
the instance shown in Fig.~\ref{fig:lbsort}:\\[-15pt]
\begin{list}{}{}
\item[(i)] There are never two free spaces, each of size
greater than or equal to $k$.
\item[(ii)]There might be more than one
free space but there is always exactly one having either size $k$ or
size $k+1$.
\end{list}
\end{lem}
\begin{proof}
(i) is obvious because otherwise the sum of the sizes of
the free spaces would exceed the total free space. (ii)
follows because in the last step either a module of size $k$ or
$k+1$ was moved leaving a free of size $k$ or $k+1$, resp.
\qed

\end{proof}

\begin{lem}
Let $ALG$ be an algorithm that uses a minimum number of steps to
sort the above instance. Then the following holds:\\[-15pt]
\begin{list}{}{}
\item[(i)] There is never more than one free space in $A$.
\item[(ii)] A module of size $k$ will only be shifted (and never be flipped).
\end{list}
\end{lem}
\begin{proof}
Consider a step that created more than one free space. This is
possible only if a module, $M_i$, of size $k$ was moved (i.e., there
is one free space of size $k$). By Lemma \ref{lem:sort1}, all other
free spaces have sizes less than $k$. Thus, only a module, $M_j$, of
size $k$ can be moved in the next step. Since we care only about the
order of the sizes of the modules not about their numbering the same
arrangement can be obtained by moving $M_j$ to the current place of
$M_i$ and omitting the flip of $M_i$ (i.e., the number of steps in
$ALG$ can be decreased); a contradiction.

From (i) we know that there is always one free space of size $k+1$
during the execution of $ALG$. Flipping a small module to this free
space creates at least two free spaces. Hence, a small module will
only be shifted. \qed
\end{proof}

\begin{thm}
Any algorithm that sorts the modules in the example from Fig.~\ref{fig:lbsort}
needs at least $\Omega(n^2)$ steps.
\end{thm}

\begin{proof}
Let $ALG$ be an algorithm that needs the minimum number of steps.
W.l.o.g.\ we assume that at the end the large modules are on the
left side of the small ones. We consider the array in its initial
configuration and, in particular, a module, $M_i$, of size $k$.
There are $\frac{i-1}{2}$ small modules, the same number of large
modules and one free space of size $k+1$ to the left of $M_i$.
Because small modules are only shifted in $ALG$ the number of small
modules on the left side of $M_i$ will not change but the number of
large ones will finally increase to $\frac{n}{2}$. Since a shift
moves $M_i$ at most a distance of $k+1$ to the right, $M_i$ has to
be shifted at least once for every large module that is moved to
$M_i$'s left. Taking the free space into account this implies that
$M_i$ has to be shifted at least $\frac{n}{2} - ( \frac{i-1}{2} + 1
)$ times, for any odd $i$ between $1$ and $n$. Hence, for $i=2j-1$
we get a lower bound of
$\sum_{j=1}^\frac{n}{2} \frac{n}{2} - j = \frac{1}{8}n^2 -
\frac{1}{4}n$
on the number of shifts in $ALG$. Additionally, every large module
has to be flipped at least once, because it has a small one to its
left in the initial configuration. This gives a lower bound of $
\frac{1}{8}n^2 - \frac{1}{4}n + \frac{1}{2}n = \frac{1}{8}n^2 +
\frac{1}{4}n $ on the total number of steps in $ALG$ and therefore a
lower bound on the number of steps for any algorithm. \qed
\end{proof}

\section{Strategies for Online Storage Allocation}
\label{sec:onlinestorageallocation} Now, we consider
the online storage allocation problem, i.e., we assume that we have
the opportunity to start with an empty array and are able to control
the placement of modules. We consider strategies that handle the
insertion and deletion of a sequence of modules. 
AlwaysSorted achieves an optimal makespan, 
possibly at the
expense of requiring up to 
$O(n^2)$ moves per insertion; the algorithm
ClassSort that is designed to require very few moves, but at the
cost of larger makespan. 
Additionally, we present a simple local heuristic, LocalShift.



\smallskip
\noindent
{\bf AlwaysSorted.}
This algorithm inserts the modules such that they are sorted
according to their size; that is, the module sizes decrease from
left to right. Note that the sorted order ensures that if a module,
$M_i$, is removed from the array all modules lying on the right side
of $M_i$ (these are at most as large as $M_i$) can be shifted $m_i$
units to the left. Now the algorithm works as follows: Before a
module, $M_j$, is inserted, we shift all modules to the left as far
as possible starting at the left side of the array. Next we search
for the position that $M_j$ should have in the array to keep the
sorted order. We shift all modules lying on the right side of the
position $m_j$ units to the right if possible; after that $M_j$ is
inserted.

\begin{thm} \label{thm:onlineinsertion}
AlwaysSorted achieves the optimal makespan.
The algorithm
performs $O(n)$ moves per insertion in the worst case.
\end{thm}

\begin{proof}
All modules are shifted to the left as far as possible before the
next module is inserted. After that, there is only one free space
at the right side of $A$. If this free space is at least as large as
the next module, the insertion is performed, meaning that a module
has to wait if and only if the total free space is smaller than the
module size; no algorithm can do better. \qed
\end{proof}

\smallskip
\noindent
{\bf DelayedSort.}
The idea is to reduce the number of moves by delaying the sorting until it is
really necessary:
We maintain a large free space on the left or the right side (alternatingly).
First, we check if we can insert the current module $M_i$, i.e.,
if $m_i \leq \sum f_j$.  Now, if we can insert $M_i$ maintaining
$\max m_i \leq \max f_j$
we insert $M_i$ using First-Fit. Otherwise, we check if $M_i$ can
be inserted---maintaining the above condition---after compacting the array using
by shifting all modules to the side where we currently keep the large
free space, beginning with the module next to the free space.
If maintaining the condition is not possible,
we sort the array using Alg.~\ref{alg:sort} and insert the module into the
single free space left after sorting.
Note that this strategy
also achieves the optimal makespan.
%
%

\smallskip
\noindent
{\bf ClassSort.}
For this strategy we assume that the size of the largest module at
most half the size of the array. We round the size of a module,
$M_i$, to the next larger power of $2$; we denote the rounded size
by $m_i'$.

We organize the array in $a=\lceil \lg \frac{|A|}{2}\rceil$ {\em
classes}, $C_0,C_1,\ldots, C_a$. Class $C_i$ has {\em level} $i$ and
stores modules of rounded size $2^i$. In addition, each class
reserves $0$, $1$, or $2$ (initially 1)
{\em buffers} for further
insertions. A buffer of level $i$ is a free space of size $2^i$. We
store the classes sorted by their level in decreasing order.

The numbers of buffers in the classes provide a sequence,
$S=s_{a},\ldots,s_0$, with $s_i\in\{0,1,2\}$. We consider this sequence as
a redundant binary number; see Brodal~\cite{b-wcepq-96}. Redundant
binary numbers use a third digit to allow additional freedom in the
representation of the counter value. More precisely, the binary
number $d_\ell d_{\ell-1}\ldots d_0$ with $d_i\in\{0,1,2\}$
represents the value $\sum_{i=0}^\ell d_i2^i$. Thus, for example,
$4_{10}$ can be represented as $100_2$, $012_2$, or $020_2$. A
redundant binary number is {\em regular}, if and only if between two
2's there is one 0, and between two 0's there is one 2. The
advantage of regular redundant binary numbers is that we can add or
subtract values of $2^k$ taking care of only $O(1)$ carries, while
usual binary numbers with $\ell$ digits and $11\ldots
1_2+1_2=100\ldots 0_2$ cause $\ell$ carries.

Inserting and deleting modules benefits from this advantage: The
reorganization of the array on insertions and deletions corresponds
to subtracting or adding, respectively, an appropriate value $2^k$
to the regular redundant binary numbers that represents the sequence
$S$. In details: If a module, $M_j$, with $m'_j= 2^i$ arrives, we
store the module in a buffer of the corresponding class $C_i$.%
\footnote{Initially, the array is empty. Thus, we create the classes
$C_1, \ldots,C_i$ if they do not already exist, reserving one free
space of size $2^k$ for every class $C_k$.} If there is no buffer
available in $C_i$, we have a carry in the counter value; that is,
we split one buffer of level $i+1$ to two buffers of level $i$;
corresponding, for example, to a transition of $\ldots 20\ldots$ to
$\ldots 12\ldots$ in the counter. Then, we subtract $2^i$ and get
$\ldots 11\ldots$. Now, the counter may be irregular; thus, we have
to change another digit. The regularity guarantees that we change
only $O(1)$ digits \cite{b-wcepq-96}. Similarly, deleting a module
with $m'_j= 2^i$ corresponds to adding $2^i$ to $S$.

\begin{thm}
ClassSort performs $O(1)$ moves per insertion or deletion in the
worst case. Let
$\hat{m}$ be the size of the largest module in the array, $c$ a
linear function and $c(m_i)$ the cost of moving a module of size
$m_i$. Then the amortized cost for inserting or deleting a module of
size $m_i$ is $O(m_i\lg\hat{m})$.
\end{thm}

\begin{proof}
The number of moves is clear. 
Now, observe a class, $C_i$. A module of size $2^i$ is moved,
if the counter of the next smaller class, $C_{i-1}$, switches from
$0$ to $2$ (for the insertion case). Because of the regular structure of
the counter, we have to insert at least modules with a total weight of
$2^{i-1}$ before we have to move a module of size $2^i$ again.
We charge the cost for this move to theses modules.
On the other hand, we charge every module at most once for every class.
As we have $\lg \hat{m})$ classes, the stated bound follows.
The same argument holds for the case of deletion. Note that
we move modules only, if the free space inside a class is not located 
on the right side of the class (for insertion) or on the left side
(for deletion). Thus,
alternatingly inserting and deleting a module of the same size
does not result in a large number
of moves, because we just imaginarily split and merge free spaces.
\qed
\end{proof}

\smallskip
\noindent
{\bf LocalShift.}
We define the distance between two blocks (modules or free spaces)
as the number of blocks that lie between these two blocks. For a
free space $F_i$ we call the set of blocks that are at most at a
distance $k \in \N$ from $F_i$ the $k$-neighborhood of $F_i$. The
algorithm LocalShift works as follows: If possible we use
BestFit to insert the next module $M_j$. Otherwise, we
look at the $k$-neighborhood of any free space (from left to right).
If shifting the modules from the $k$-neighborhood, lying on the left
side of $F_i$, to the left as far as possible (starting a the left
side) and the modules lying on the right side to the right as far as
possible (starting at the right side) would create a free space that
is at least as large as $M_j$ we actually perform these shifts and
insert $M_j$. If no such free space can be created, $M_j$ has to
wait until at least one modules is removed from the array. This
algorithm performs at most $2k$ moves per insertion.

\begin{figure}[t]
\mbox{}\hspace{-15pt}%
{\epsfig{figure=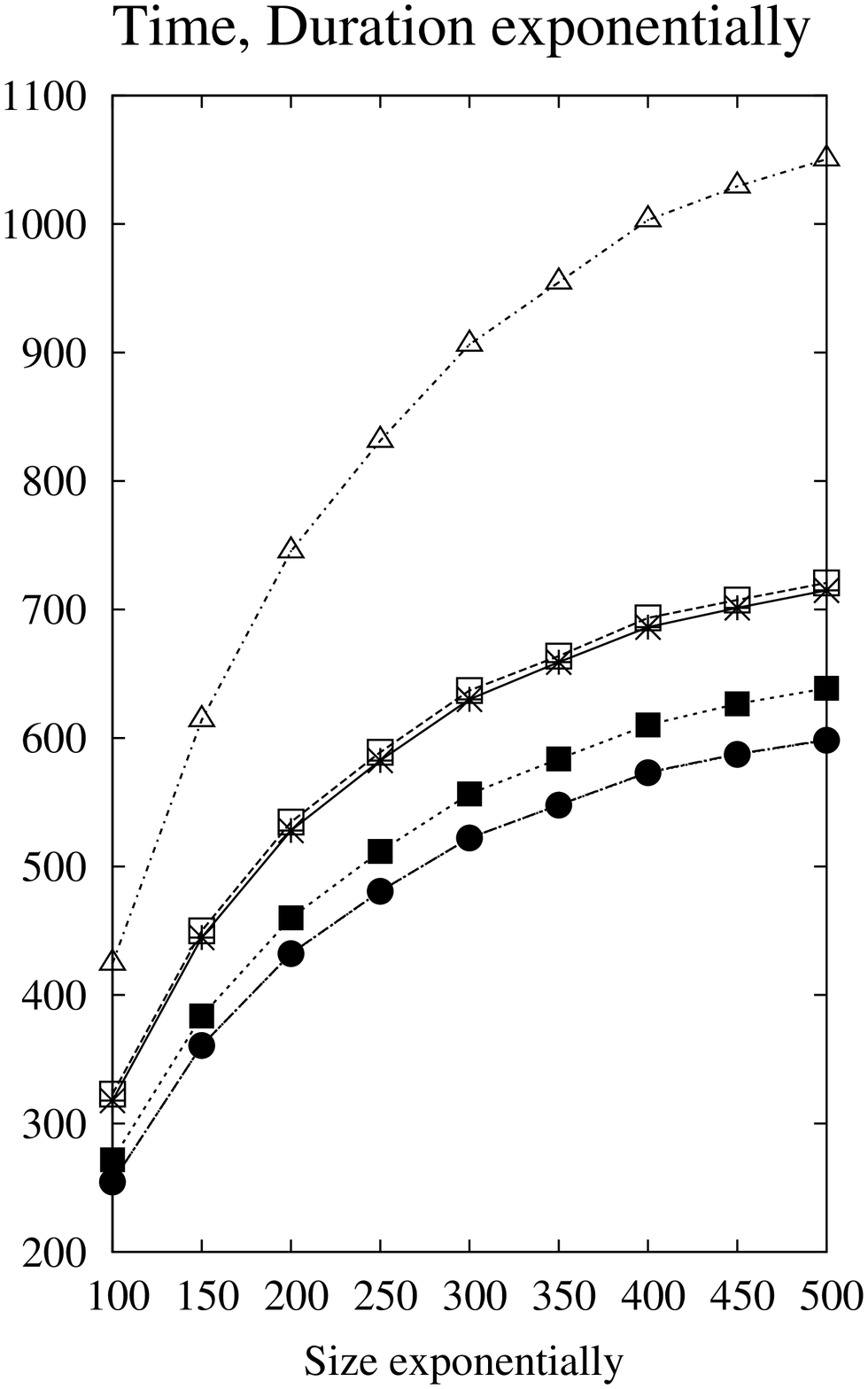,width=45mm}}\hspace{-4mm}%
{\epsfig{figure=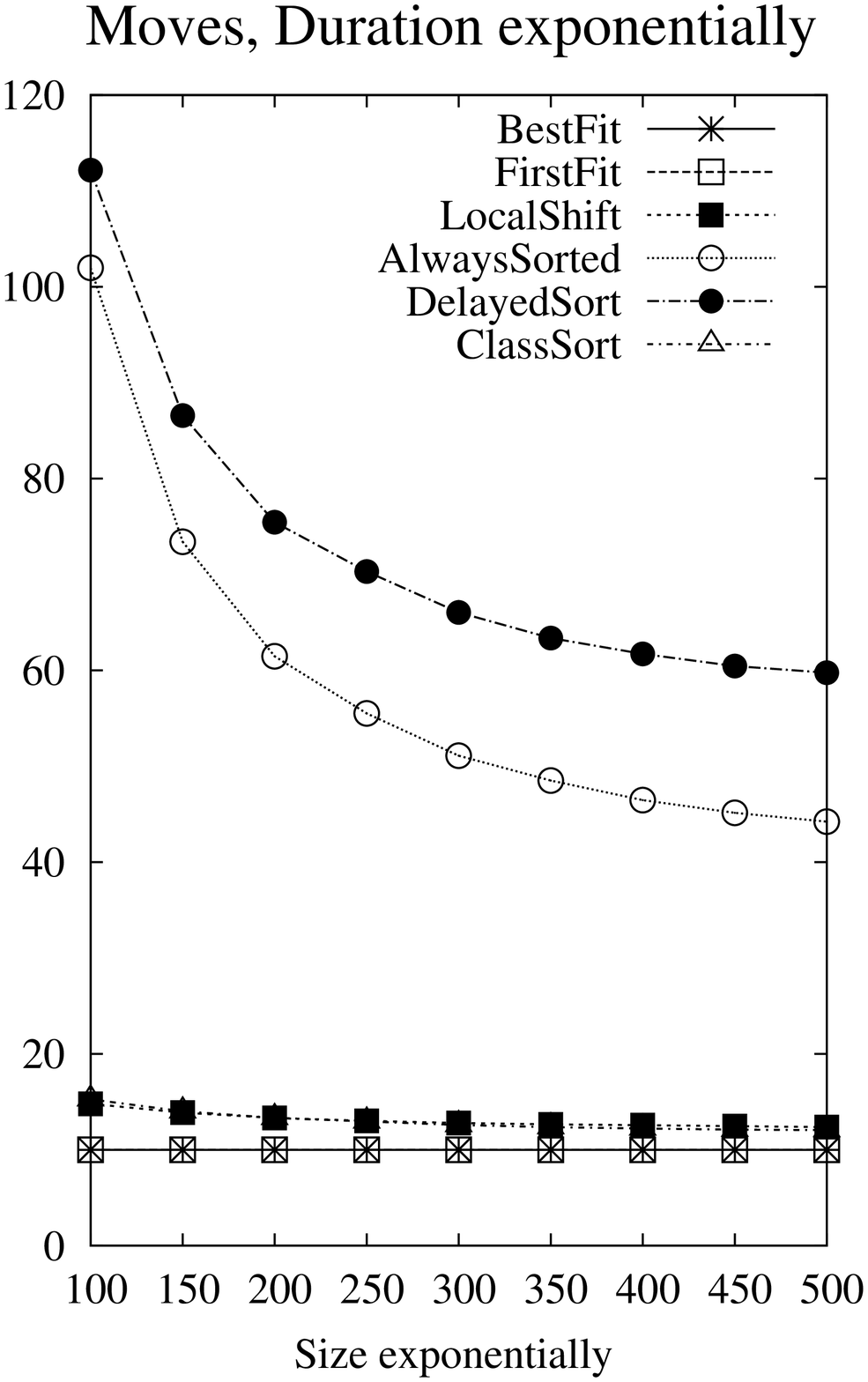,width=45mm}}\hspace{-4mm}%
{\epsfig{figure=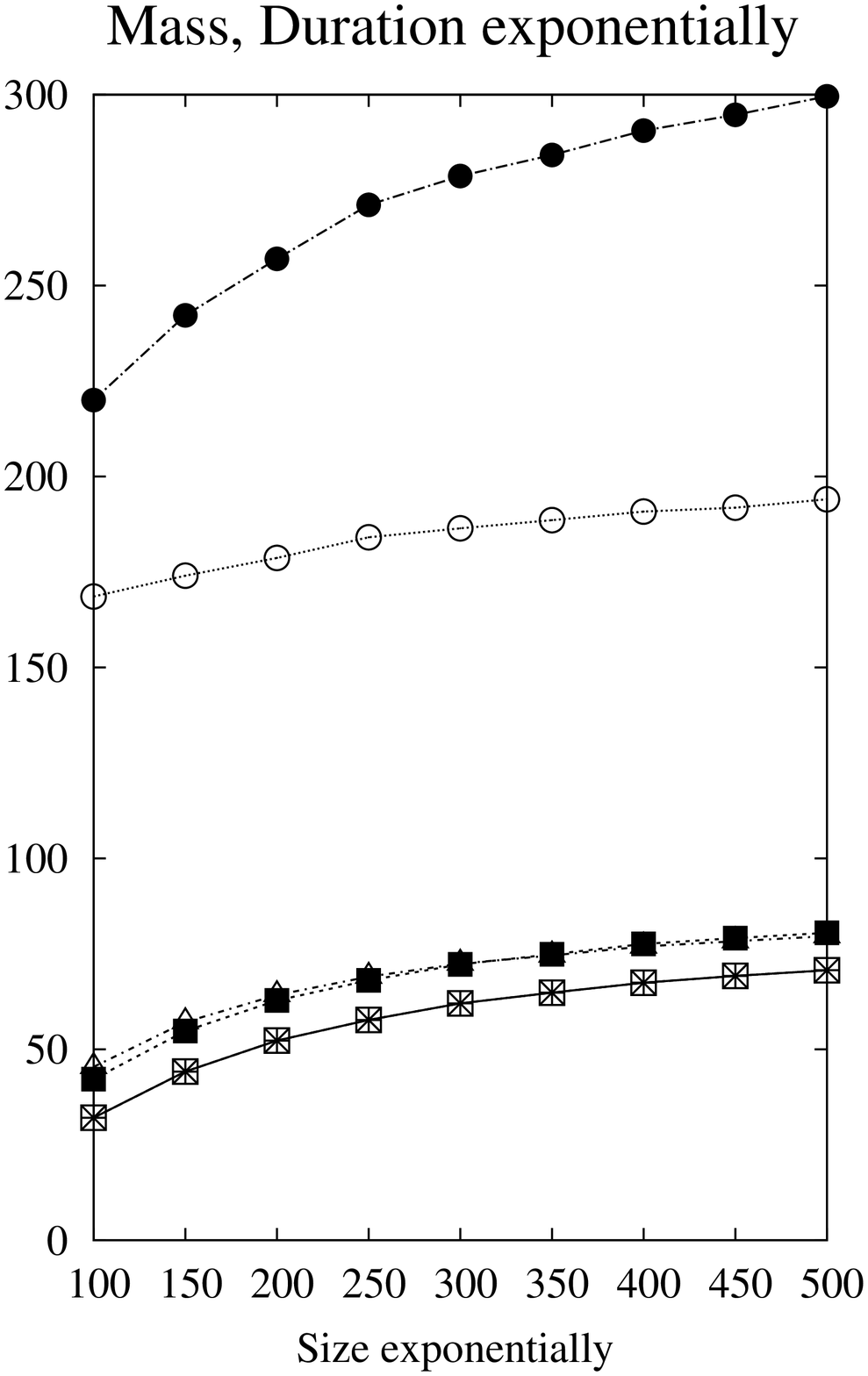,width=45mm}}%
\caption{Experiments with exponential distribution for size and duration.
\label{fig:results1}\vspace{-5mm}}
\end{figure}

\section{Comparison and Conclusion}
To test our strategies, we generated a number of random input sequences 
and analyzed the performance of our strategies 
as well as the simple FirstFit and BestFit approaches
in an array of size $2^{10}$. 
A sequence consists of 100,000 modules, each module has a
randomly chosen size and duration time. 
For each sequence, size and time are shuffled using several
probability distributions.
We analyzed three objectives: the time to complete the whole sequence 
(the makespan),
the number of moved modules ($c(m_i)=1$) and the moved mass ($c(m_i)=m_i$).
Our experiments (see Fig.~\ref{fig:results1} for an example)
showed that LocalShift
performs very well, as it constitutes a compromise
between a moderate number moves and a low makespan.
Both makespan and moves turn out to be nearly optimal.

The more complex strategy ClassSort performed only slightly worse than
LocalShift concerning moves, but disappoints in its resulting makespan.
In contrast, both types of sorting-related strategies
have---of course---a good makespan, but need a lot of moves.
Unsurprisingly, FirstFit and BestFit need the fewest moves (as
they perform moves only on inserting a module, but never move a
previously placed module). Their makespan turned out to be clearly better than
ClassSort, but worse than LocalShift and the sorting strategies.

A comparison of the sorting strategies, AlwaysSorted and
DelayedSort, showed that delaying the sorting of the array until
it is really necessary pays off for the number of moves, but not if we count
the moved mass, this is because the shift from 
maintaining one large free space 
to sorting (caused by not enough free space to accompany the largest
item) results in a sequence with several moves of the heaviest items,
which is not the case for AlwaysSorted.

We have introduced the systematic study of dynamic storage
allocation for contiguous objects. 
There are still
a number of open questions, such as the worst-case number of moves required to
achieve connected free space or cheaper certificates for guaranteeing that connected 
free space can be achieved.


\end{document}

%% file: npcomplete-scaled.pstex_t
\begin{picture}(0,0)%
\includegraphics{npcomplete-scaled.pstex}%
\end{picture}%
\setlength{\unitlength}{2072sp}%
\begingroup\makeatletter\ifx\SetFigFont\undefined%
\gdef\SetFigFont#1#2#3#4#5{%
  \reset@font\fontsize{#1}{#2pt}%
  \fontfamily{#3}\fontseries{#4}\fontshape{#5}%
  \selectfont}%
\fi\endgroup%
\begin{picture}(7903,1190)(895,-3241)
\put(8416,-2536){\makebox(0,0)[b]{\smash{{\SetFigFont{5}{6.0}{\rmdefault}{\mddefault}{\updefault}{\color[rgb]{0,0,0}$M_{4k+1}$}%
}}}}
\put(7201,-2536){\makebox(0,0)[b]{\smash{{\SetFigFont{5}{6.0}{\rmdefault}{\mddefault}{\updefault}{\color[rgb]{0,0,0}$M_{4k}$}%
}}}}
\put(5491,-2536){\makebox(0,0)[b]{\smash{{\SetFigFont{5}{6.0}{\rmdefault}{\mddefault}{\updefault}{\color[rgb]{0,0,0}$M_{3k+3}$}%
}}}}
\put(4276,-2536){\makebox(0,0)[b]{\smash{{\SetFigFont{5}{6.0}{\rmdefault}{\mddefault}{\updefault}{\color[rgb]{0,0,0}$M_{3k+2}$}%
}}}}
\put(1846,-2536){\makebox(0,0)[b]{\smash{{\SetFigFont{5}{6.0}{\rmdefault}{\mddefault}{\updefault}{\color[rgb]{0,0,0}$M_3$}%
}}}}
\put(1486,-2536){\makebox(0,0)[b]{\smash{{\SetFigFont{5}{6.0}{\rmdefault}{\mddefault}{\updefault}{\color[rgb]{0,0,0}$M_2$}%
}}}}
\put(1126,-2536){\makebox(0,0)[b]{\smash{{\SetFigFont{5}{6.0}{\rmdefault}{\mddefault}{\updefault}{\color[rgb]{0,0,0}$M_1$}%
}}}}
\put(3196,-2536){\makebox(0,0)[b]{\smash{{\SetFigFont{5}{6.0}{\rmdefault}{\mddefault}{\updefault}{\color[rgb]{0,0,0}$M_{3k+1}$}%
}}}}
\put(2566,-2536){\makebox(0,0)[b]{\smash{{\SetFigFont{5}{6.0}{\rmdefault}{\mddefault}{\updefault}{\color[rgb]{0,0,0}$M_{3k}$}%
}}}}
\put(1825,-2210){\makebox(0,0)[b]{\smash{{\SetFigFont{6}{7.2}{\familydefault}{\mddefault}{\updefault}{\color[rgb]{0,0,0}$kB$}%
}}}}
\put(5520,-2210){\makebox(0,0)[b]{\smash{{\SetFigFont{6}{7.2}{\familydefault}{\mddefault}{\updefault}{\color[rgb]{0,0,0}$kB+1$}%
}}}}
\put(3673,-2210){\makebox(0,0)[b]{\smash{{\SetFigFont{6}{7.2}{\familydefault}{\mddefault}{\updefault}{\color[rgb]{0,0,0}$B$}%
}}}}
\end{picture}%